\documentclass[10pt, letterpaper, twocolumn]{article}

\usepackage[utf8]{inputenc}
\usepackage[english]{babel}
\usepackage{geometry}
\usepackage{setspace}
\usepackage{amsmath}
\usepackage{amssymb}
\usepackage{amsfonts}
\usepackage{amsthm}
\usepackage{mathtools}
\usepackage{algorithm}
\usepackage{tikz}
\usepackage[noend]{algpseudocode}
\usepackage{lipsum}
\usepackage[backend=biber, style=numeric]{biblatex}
\usepackage{changepage}

\usepackage{url}

\usepackage{breakurl}
\usepackage[breaklinks]{hyperref}

\newtheorem{theorem}{Theorem}
\newtheorem{corollary}{Corollary}
\newtheorem{lemma}{Lemma}
\newtheorem{proposition}{Proposition}

\theoremstyle{remark}
\newtheorem*{remark}{Remark}

\theoremstyle{definition}
\newtheorem{definition}{Definition}

\title{Sorting an Array Using the Topological Sort \\ of a Corresponding Comparison Graph}
\date{June 2020}

\begin{document}

\sloppy

\twocolumn[



\begin{center}
    \LARGE{\bf Sorting an Array Using the Topological Sort \\ of a Corresponding Comparison Graph}

    \vspace{1.5em}

    \large{\bf June 2020}

    \vspace{1.5em}

    \normalsize{\textbf{Balaram D. Behera}}
    \vspace{0.25em} \\
    \normalsize{\texttt{bbehera@ucsc.edu}}
    \vspace{0.25em} \\
    \normalsize{Computer Science \&} \\ \normalsize{Engineering Department}
    \vspace{0.25em} \\
    \normalsize{University of California} \\
    \normalsize{Santa Cruz}

    \vspace{1em}
\end{center}

\begin{center}
    \subsection*{Abstract}
\end{center}
\begin{adjustwidth}{55pt}{55pt}
    \small{The quest for efficient sorting is ongoing, and we will explore a graph-based stable sorting strategy, in particular employing comparison graphs.
    We use the topological sort to map the comparison graph to a linear domain, and we can manipulate our graph such that the resulting topological sort is the sorted array.
    By taking advantage of the many relations between Hamiltonian paths and topological sorts in comparison graphs, we design a Divide-and-Conquer algorithm that runs in the optimal $O(n \log n)$ time.
    In the process, we construct a new merge process for graphs with relevant invariant properties for our use.
    Furthermore, this method is more space-efficient than the famous {\sc MergeSort} since we modify our fixed graph only.}

    \noindent \emph{Keywords: Graph Algorithms; Topological Sort; Sorting Algorithms; Comparison Graphs.}
\end{adjustwidth}

\vspace{2em}
]

\section{Introduction}
Most sorting algorithms run a multitude of array comparisons, and from those results we decide how to manipulate the elements to eventually achieve a sorted ordering of elements.
This process can be implemented and infused with all kinds of data structures.
In particular, directed graphs are a great way to structure this problem, and this allows us to look at sorting in different light and realize new methods of sorting.

We can represent every element as a vertex and the result of every comparison as an arc.
Thus we can construct a graph that essentially stores all the comparisons made.
In fact, we can mathematically represent these comparisons as an order relation: construct an arc if and only if the origin is less than the terminus (in the case of distinct array elements). Now we must decipher such a graph, i.e. to find some meaning to all those arcs that we have created.
Our end goal is to achieve a sort of our input array, and in this paper we plan to achieve this using the topological sort of the graph (defined later).

We will employ basic ideas of graph theory, the Depth-First Search algorithm, and the topological sort to efficiently sort an array using directed graphs.
We will first explore a somewhat trivial way to solve this problem, and then build a more efficient algorithm that will give us an equivalent result.

\section{Previous Work}
Rajat K. Pal in his first paper on this topic discussed his {\sc RKPianGraphSort} algorithm which was a perfect graph-based sorting algorithm which was more extensive than what we present here and it ran in time $O(n^2)$ \cite{rkp1}. In his second paper \cite{rkp2}, Pal introduced a complete graph-based sorting algorithm which we develop here, calling it the Trivial Algorithm, yet this algorithm also runs in $O(n^2)$ time \cite{rkp2}.

The main limitation in previous work has been the complete graph construction required to employ the topological sort effectively, so in this paper we counter that with our intermediary merge processes to enhance the run time to reach the optimal $O(n \log n)$.
In the buildup to the novel {\sc GraphSort} algorithm, we go through the {\sc TrivialGraphSort} algorithm which has been expounded by Pal \cite{rkp2}.
This assists in the building of the required theory which has a different approach to Pal's as we go in the direction of Hamiltonian paths imperative in the development of the improved algorithm.

\section{Definitions}
Let us define a few terms that will be used frequently throughout the course of this paper.
First let's define the family of graphs we specifically are working with.
\begin{definition}
    Let $G = (V, E)$ be a directed graph with vertex set $V$ and edge set $E$.
    Let $V = \{1, \ldots, n\}$ where $n$ is the order.
    Let the following set $A = \{a_1, \ldots, a_n\}$ represent the values of their corresponding vertex.
    Let the edge $(u, v) \in E$ for vertices $u$ and $v$ if and only if $a_u < a_v$.
    Then the graph $G$ is a comparison graph.
\end{definition}
\begin{remark}
    We often consider the corresponding set $A$ as the given array itself.
\end{remark}

Let's define the topological sort which is the basis of this entire paper.
\begin{definition}
    Let $G$ be a directed acylic graph with order $n$.
    Let $S = (s_1, \ldots, s_n)$ be a sequence of all vertices such that for all $1 \leq i \leq n$, the vertex $s_i$ is not adjacent to vertices $s_k$ such that $1 \leq k \leq i$.
    Then the sequence $S$ is a topological sort of graph $G$.
\end{definition}
\begin{remark}
    The topological sort of a graph is not necessarily unique.
\end{remark}

A term we will use to evaluate how close we are to achieving a directed acyclic graph with a unique topological sort is trueness.
\begin{definition}
    Let $G$ be a directed acyclic graph, and let $S$ represent a topological sort of $G$.
    The number of elements in $S$ that are not fixed, i.e. their immediate adjacent neighbors (previous and next element) are not unique, is the trueness of $G$, denoted by $\tau(G)$.
\end{definition}
\begin{remark}
    Note that the trueness is not a count of how many topological sorts (that's more of a combinatorial extension), rather we count the number of elements in the sequence that are not fixed.
\end{remark}
\begin{remark}
    A graph with only one topological sort has a trueness of one, i.e. $\tau(G) = 1$ as all vertices are fixed.
\end{remark}
\begin{remark}
    In general, we focus on improving the trueness, i.e. to decrease $\tau(G)$ rather than evaluate it.
\end{remark}

Since we will be working with arrays and array comparisons, we may want to compare an element of an array with some of its neighbors, so we have the following definition.
\begin{definition}
    Let $A$ be an array of size $n$, and let $1 \leq i \leq n$.
    Let $r$ be the comparison reach (or just reach) of $A$.
    Then we compare $A[i]$ only to $A[i - r], \ldots, A[i - 1]$ and $A[i + 1], \ldots, A[i + r]$.
    Note if an array index is out of bounds, we wrap it around the array, i.e. $A[-1]$ means $A[n - 1]$.
\end{definition}
\begin{remark}
    In other words, we wish to compare an element $A[i]$ to the closest $2r$ neighbors (equally distributed on both ``sides'').
\end{remark}

Now we define a comparison graph called the corresponding graph that is a graph constructed from an array and some chosen array comparisons.
\begin{definition}
    Let $G$ be a null directed graph, i.e. $V = \{1, \ldots, n\}$ and $E = \varnothing$.
    Let $A$ be an array, and $r$ the comparison reach of $A$.
    Then add $(i, j)$ for all vertices $i$ and $j$ such that $A[i] < A[j]$ and such that $j \in [i - r, i + r]$ but $j \neq i$.
    Then $G$ is a corresponding graph of $A$ of reach $r$.
\end{definition}
\begin{remark}
    A corresponding graph is a comparison graph, where our value set $A$ is the corresponding array.
    The corresponding graph is simply a specific structure of comparison graphs.
\end{remark}
\begin{remark}
    The initial ordering of the array can generate different corresponding graphs, and different reach values also may generate different corresponding graphs.
\end{remark}

Lastly, we define the directed Hamiltonian path which is well-known, but we will repeat the definition to aid with upcoming theorems.
\begin{definition}
    Let $G$ be a directed graph.
    Let $P$ be a path such that $P$ includes all the vertices of $G$.
    Then $P$ is called a directed Hamiltonian path in $G$.
\end{definition}

These definitions will come in play during the development of theory in this paper.
Moreover, we will frequently use basic graph theory definitions of trees, forests, connectedness, acyclicity, etc. which the reader is assumed to know.

\section{Preliminaries}
Before we start diving into the algorithms let's note some important preliminary remarks, assumptions, theorems, data structures, and algorithms that will be used in this paper.

Concerning the reach of an array, we will mainly study algorithms with a reach of one or for the sake of simplicity.
Constructing corresponding graphs with higher reach values definitely generates more complex graphs, but as we will see later, the algorithms concerning those values may be more efficient.
Also, we will study the total reach which is a reach value equal to the length of the array, and this will be used for the trivial algorithm.

Throughout the paper all the theorems stated apply to both comparison and corresponding graphs in most cases.
The only difference between the two is that the corresponding graph has a defined reach whereas the comparison graph has no such defined structure.

It's important to note that for the sake of proving correctness, we assume all arrays have distinct values (later we will discuss what happens when this is not the case).

Since we will be using the topological sort of graphs, let's first show that the graphs we are working with in fact do have topological sorts.
\begin{theorem}
    If $G$ is a comparison graph, then $G$ is acyclic.
    Hence $G$ is a directed acyclic graph.
\end{theorem}
\begin{proof}
    Let $G$ be a comparison graph of order $n$.
    We must show that $G$ is acylic.
    To the contrary, assume $G$ contains a cycle, call it $C$.
    Then for $k \leq n$ as the length of $C$, let $C = (1, \ldots, k, 1)$ where $i$ (for $1 \leq i \leq k$) is a vertex in $G$ and every two adjacent vertices in $C$ is a directed arc in $C$.
    Let $a_i$ also represent the corresponding value for every vertex in $G$.
    Then by the definition of a comparison graph, we have
    \begin{align*}
        a_1 < \ldots a_k < a_1
    \end{align*}
    which is an obvious contradiction.
    Therefore, $G$ is acyclic, i.e. $G$ is a directed acylic graph.
\end{proof}
We can now easily derive the following corollaries.
\begin{corollary}
    If $G$ is a comparison graph, then $G$ has a topological sort.
\end{corollary}
\begin{proof}
    Let $G$ be a comparison graph.
    By Theorem 1, $G$ is a directed acylic graph.
    Then by definition, $G$ has a topological sort.
\end{proof}
\begin{corollary}
    If $G$ is a corresponding graph of array $A$ with any valid reach $r$, then $G$ is acylic and has a topological sort.
\end{corollary}
\begin{proof}
    Since every corresponding graph is a comparison graph with the value set being the array $A$, the graph $G$ necessarily has a topological sort following immediately from Theorem 1 and Corollary 1.
\end{proof}

Although we can construct algorithms dealing simply with the abstract mathematical definitions of graphs, to implement these algorithms in practice, it's imperative to construct useful data structures to represent graphs.
We will use the adjacency list representation of graphs where every vertex has a list of vertices adjacent to it.
This data structure is very memory efficient and easily accessible.
In fact, we will be choosing certain orientations of this list to best cater to our needs in our procedures.

We will be employing Depth-First Search (DFS) all throughout this paper.
It is assumed the reader has firm knowledge of this procedure and understands its applications.
In summary, DFS means to travel deep from a root vertex till we can only travel to a vertex that has already been visited which is when we start back-tracking to parent vertices to continue on a different traversal.
Once we back-track to the root itself, we have discovered that part of the graph and we do the same process for some other unvisited root.
In this way we traverse the entire graph.
The run-time of DFS on some graph $G$ is $O(n + m)$ where $n$ and $m$ are order and size of $G$, respectively.
While running DFS we keep track of parent vertices and start and finish times, where the start time is the discrete time at which we visit a vertex and a finish time is the discrete time at which we finish visiting all its children.
The following theorem is the main reason DFS is relevant to our purpose.
\begin{theorem}
    Let $G$ be a directed acyclic graph.
    Run Depth-First Search (DFS) on $G$, and generate a stack $S$ where the top has the greatest finish time and the bottom has the lowest finish time (ordered by decreasing finish times).
    Then that stack $S$ is a topological sort of $G$ \cite{clrs}.
\end{theorem}
We will take this theorem as is without proof \cite{clrs}.
This theorem is the backbone of many of the algorithms we will study further in this paper.

Lastly, we will call the components of a directed graph as the components of the underlying graph.
Hence, there is no conclusion about weak or strong connectedness involved.

\section{Construction of the Corresponding Graph}
The basis of sorting using graphs is to first construct the graph.
The definition of the corresponding graph given before essentially lays out our method of construction.
Note that algorithms may vary immensely depending on the reach, so we fix the comparison reach for every algorithm.
However, to preserve generality, consider $r$ to be the reach of an array $A$ of $n$ elements.
We will use the adjacency list representation for our corresponding graph $G$ of reach $r$ on $A$.
By definition, we will be comparing $A[i]$ with $A[i + k]$ for all indices $1 \leq i \leq n$ and all $k$ in the range $1 \leq k \leq r$ (not if $i + k > n$, then we simple take $i + k$ modulo $n$).
In this way we will loop through every single index $i$ and make the necessary comparisons.
Note that although a reach $r$ denotes that we compare an element to the closest $r$ neighbors to the left and $r$ neighbors to the right, it is equivalent in practice to simply compare just the $r$ right neighbors or the $r$ left neighbors.

However, before we implement a process to construct the graph, it's important to discuss the process of adding an arc.
We use the adjacency list operation, and we can choose for the list to be sorted, reverse sorted, or completely arbitrary (similar to a set).
By sorted, we mean with respect to the given array values.
Of course it benefits us that it is sorted, since we will see that we normally want to traverse to the next smallest valued vertex.
We achieve this sorted invariant by inserting every new vertex in the correct position in our adjacency list so it remains sorted.

There are two strategies for this: a binary search inspired insert or an insertion sort inspired insert.
For the binary search inspired algorithm, the basic idea is that we check the middle of the list and determine which half the new vertex should be in, and in this way we continue till there is only one spot.
This has the condition that it requires a fast way to access adjacency list elements.
On the other hand, the insertion sort inspired algorithm where the basic idea is to go through the list from the first to last element and find where the new vertex should fit.
This algorithm is more versatile and compatible with the general list data structure, so use this.
Moreover, the asymptotic run-time for our algorithms is not affected by this choice.

\subsection{Pseudo-code}
Let's first see the process of adding an arc using the insertion sort inspired algorithm (note this is completely independent of our construction process, it simply acts as a method of our graph data structure).
Let $G$ be our graph, $A$ the array that has the values of the corresponding graph $G$, and we want to add the arc $(u, v)$.
Our adjacency list is represented by an array of doubly-linked lists $\mathrm{adj[]}$ with general double-linked list and array functionality.
The following is the pseudo-code for this process.
\begin{algorithmic}
    \small
    \Function {AddArc}{$G, A, u, v$}
        \State $w = \mathrm{adj}[u].\mathrm{front}$
        \While {$\mathrm{adj}[u].\mathrm{index} > 0$}
            \If {$A[v] < A[w]$}
                \State \Call{Insert}{$\mathrm{adj}[u], w.\mathrm{index}, v$}
                \State $G.\mathrm{size}++$
                \State \Return
            \EndIf
            \If {$\mathrm{adj}[u].\mathrm{index} = \mathrm{adj}[u].\mathrm{length}$}
                \State \Call{Append}{$\mathrm{adj}[u], v$}
                \State $G.\mathrm{size}++$
            \EndIf
            \State $w = w.\mathrm{next}$
        \EndWhile
    \EndFunction
\end{algorithmic}
This particular function will keep the adjacency list sorted in ascending order of values.
This configuration can vary depending on what is required for a parent process.

To construct the graph, we will add an arc depending on the value of the two elements being compared.
Let $A$ be an array of $n$ elements, and $r$ the comparison reach for the corresponding graph $G$ of $A$.
The following is the pseudo-code of the algorithm described for the construction.
\begin{algorithmic}
    \small
    \Function {ConstructGraph}{$A, n, r$}
        \State $G = $ null graph of order $n$
        \For {$i = 1$ to $n$}
            \For {$k = 1$ to $r$}
                \State $j = i + k \mod n$
                \If {$A[i] < A[j]$}
                    \State \Call {AddArc}{$G, A, i, j$}
                \Else
                    \State \Call {AddArc}{$G, A, j, i$}
                \EndIf
            \EndFor
        \EndFor
        \State \Return $G$
    \EndFunction
\end{algorithmic}
Note that we assume $A$ contains distinctly valued elements.
Further this process is quite simple as it always does the required number of comparisons.

\subsection{Run-time}
For adding an edge, the worst case is when we must go through the entire adjacency list and then choose to add the vertex to the end, since we need to compare the new vertex with all the other elements.
Let $a$ be the maximum length of the adjacency list that we wish to add to, then our worst-case run-time for adding an edge is $\Theta(a)$.
Note that normally the worst run-time is $\Theta(n)$ where $n$ is the order of graph, but for our sake we will keep the run-time in terms of $a$.

The construction of the graph is also a very simple process.
By definition, every vertex can only be adjacent to a maximum of $2r$ vertices since that is our comparison reach.
Thus, we let $a = 2r$, i.e. maximum length of any adjacency list.
Since we use the compare-right-only method, for every element we add $r$ arcs.
Lastly, we have $n$ elements, and for each one we follow this method.
In summary, we have a naive run-time of $O(n \cdot r \cdot 2r)$ which equates to $O(nr^2)$ which is quadratic in terms of $r$ and linear in terms of $n$.

\section{Properties of the Topological Sort}
Before we discuss the core algorithms, let's first explore the properties and applicability of the topological sort for our sorting problem.
Firstly, let's show that the existence of the topological sort of the corresponding graph being an actual sort of the array.
\begin{theorem}
    Let $G$ be a corresponding graph of the array $A$.
    Then there exists a topological sort of $G$ such that it is the sequence of indices of the sorted array of $A$.
\end{theorem}
\begin{proof} 
    Firstly, there exists at least one topological sort in $G$ by Corollary 2.
    Let $A'$ be the sorted array of $A$.
    Let $S$ be the sequence of indices in $A$ such that the $A$-values for the these indices in $A'$ are sorted.
    Then our claim is $S$ is a topological sort of $G$.
    Assume to the contrary, that $S$ is not a topological sort, i.e. it contains a back-edge $(i, j)$ such that $i > j$ where $i, j \in S$, i.e. indices of $A$.
    By definition of a corresponding graph, if there exists an arc $(i, j)$, then $A[i] < A[j]$, then since $A'$ is sorted, we conclude $i < j$ which contradicts our hypothesis.
    Hence, there exists such a topological sort of $G$.
\end{proof}

We have the following immediate corollary that will be important in some time.
\begin{corollary}
    If there exists precisely one topological sort of the corresponding graph $G$ of the array $A$, then that topological sort is the sequence of indices of the sorted array of $A$.
\end{corollary}
\begin{proof}
    From Theorem 3, we know there must always exist a topological sort of $G$ that is the sequence of indices from the sorted array of $A$.
    Since there is only one topological sort, it must be such a sequence.
\end{proof}

Now let's work toward finding a case where there is only one topological sort.
\begin{lemma}
    There exists at most one Hamiltonian path, i.e. a path that visits every vertex, in any comparison graph (includes corresponding graphs).
\end{lemma}
\begin{proof}
    By way of contradiction, assume there exist two distinct Hamiltonian paths in comparison graph $G$ called $H$ and $H'$.
    Then there exists an arc $e = uw \in E(H) - E(H')$ for vertices $u$ and $v$.
    Then there must exist an edge $e' = vw \in E(H') - E(H)$ for vertices $v$ and $w$ since every Hamiltonian path must have an arc for every vertex.
    Thus $A[u] < A[w]$ and $A[v] < A[w]$.
    Without loss of generality, assume $A[u] < A[v]$.
    Note the array $A$ represents the $A$-values of $G$.
    Then in $H$, we cannot visit $v$ after visiting $w$ since $A[v] < A[w]$ and all vertices after $w$ in $H$ also have greater value than $v$.
    Thus we must visit $v$ before we traverse the edge $e$ which follows that we cannot visit $u$ again with the same reasoning.
    This contradicts the hypothesis that we visit $v$ before $u$.
    Hence there can exist at most one Hamiltonian path in $G$.
\end{proof}

Now we have enough material to show a relationship between a Hamiltonian path and the topological sort of some corresponding graph.
\begin{theorem}
    If there exists a Hamiltonian path in the comparison graph $G$ (includes corresponding graphs), then there exists precisely one topological sort of $G$.
\end{theorem}
\begin{proof}
    Assume there exists a Hamiltonian path $H = (v_1, \ldots, v_n)$ in the comparison graph $G$ with corresponding values from array $A$.
    Let $S = (v_1, \ldots, v_n)$, i.e. the vertices in $H$ in the same order.
    Our claim is that $S$ is a topological sort.
    By way of contradiction, assume not, i.e. there exists an arc $e = (v_i, v_j)$ where $i > j$.
    In $H$ there is a path from $v_j$ to $v_i$ since $j < i$ by definition of a Hamiltonian path and a corresponding graph.
    Thus the arc $e$ produces a cycle which contradicts Corollary 2.
    Hence, there exists a topological sort $S$ of $G$ if the graph contains a Hamiltonian path.

    Now to show uniqueness, by way of contradiction, assume there exists another topological sort
    \begin{align*}
        S' = (v_1, \ldots, v_k, u_{k + 1}, \ldots, u_n)
    \end{align*}
    distinct from $S$ where $u_{k + 1}$ is the first element that is different between the two sequences.
    So the value of $u_{k + 1}$ is greater than $v_{k + 1}$ because otherwise $S$ would contain a back-edge.
    Then, $u_i \ne v_{k + 1}$ for all $k + 1 \leq i \leq n$, so $v_{k + 1} \not\in S'$ which contradicts our assumption of $S'$ a topological sort since it doesn't contain a vertex.
    Therefore, there exists a unique, i.e. precisely one, topological sort of $G$ as claimed.
\end{proof}

Thus we can show the big theorem for this section whose conditions we wish to satisfy through our algorithms.
\begin{theorem}
    If there exists a Hamiltonian path in the comparison graph $G$ (includes corresponding graphs) with corresponding values from array $A$, then the topological sort of $G$ is the sequence of indices of the sorted array of $A$.
\end{theorem}
\begin{proof} 
    By Theorem 4, there is precisely one topological sort of $G$, and by Corollary 3 that topological sort is the desired sequence of indices as claimed.
\end{proof}

A Hamiltonian path in the comparison or corresponding graph implies that we have an arc from a vertex to the smallest greater value vertex, i.e. we have an arc between adjacent elements in the sorted array.
Hence, we must either get to the point where we have a Hamiltonian path in the corresponding graph which makes it very easy to find the sorted array, or we develop a method to find all the topological sorts and find which one is the sorted array.

This paper will delve into the former strategy as we wish to use DFS.
The DFS algorithm is integral to the first strategy because it gives an efficient method to finding the topological sort, and since there should be only one, we are done after running DFS.

\subsection{Run-time of DFS}
When we introduced DFS we mentioned its asymptotic run-time is $O(n + m)$ where $n$ and $m$ are the order and size of the graph $G$ on which we run DFS on.
Let $G$ be a corresponding graph of array $A$ with $n$ elements and a reach of $r$.
Then the order for $G$ is also $n$.
We know every vertex will be adjacent to or from precisely $2r$ vertices by definition of the reach.
Thus by the degree-sum formula:
\begin{align*}
    &\quad \sum_{v \in V(G)} \mathrm{deg}(v) = n(2r) = 2m \\
    &\therefore m = nr.
\end{align*}
Thus the run-time of DFS on $G$, a corresponding graph, is $O(n + nr)$ which equates to $O(nr)$ since $r$ is greater than one.

However it is important to note that if any comparison graph contains a Hamiltonian path and we start DFS at the minimum value vertex, we only have a run-time of $\Theta(n)$ for DFS which is the time it takes to traverse the Hamiltonian path after which we have visited every vertex as required.

\section{A Trivial Algorithm}
From Theorem 5 we know that if we can construct a corresponding graph such that it contains a Hamiltonian path, it is a trivial process to find the sorted array.
The simplest method to achieve this is to construct the most complete corresponding graph, i.e. we compare every element of an array $A$ with every other element of $A$.
In this way, our graph necessarily must contain a Hamiltonian path; then we trivially we can determine the sorted array.
So let $G$ be a corresponding graph of array $A$ of $n$ elements with a reach of $r = \lfloor n / 2 \rfloor$, i.e. we compare every element of $A$ with the closest $2r \geq n - 1$ neighbors (there are $n - 1$ elements other than the current).
This graph $G$ is the most complete corresponding graph of $A$ since we cannot add more arcs.
Now we can run DFS on $G$ starting at the minimum value element of $A$.
Then we claim that the DFS stack is the sorted array of $A$.

\subsection{Pseudo-code}
Let $A$ be an array of $n$ elements.
Let \textsc{FindMin($A$)} return the index of the minimum value element of $A$ that runs in linear time.
Secondly let \textsc{DFS($A, S$)} be the function that runs the DFS algorithm where $S$ is the order by which DFS does its uppermost visits, and after the process is done $S$ is the DFS stack.
For us we will have $S = (x)$ where $x$ is the index of the minimum value element.
Lastly, the procedure \textsc{ToArray($A, S$)} constructs an array from a list $S$ of indices of $A$.
Then we have the following pseudo-code for the trivial algorithm.
\begin{algorithmic}
    \Function {TrivialGraphSort}{$A$}
        \State $n = A.\mathrm{length}$
        \State $G = $ \Call {ConstructGraph}{$A, n, \lfloor n / 2 \rfloor$}
        \State $x = $ \Call {FindMin}{$A$}
        \State $S = (x)$
        \State \Call {DFS}{$G, S$}
        \State \Return \Call{ToArray}{$A, S$}
    \EndFunction
\end{algorithmic}
Note that $S$ is a stack, so that must be converted to an array technically.
The trivial algorithm is not very complex as it simply fulfills the requirements of Theorem 5 to achieve the sort with the most naive strategy.

\subsection{Correctness}
Before we delve into proving the correctness of the algorithm, let's define a term we have casually used.
\begin{definition}
    Let $G$ be a corresponding graph.
    If we can't add more arcs to $G$, it is a complete corresponding graph, i.e. it is maximal with respect to arcs.
\end{definition}

Now let's prove an important statement related to complete corresponding graphs.
\begin{theorem}
    Every complete corresponding graph $G$ of array $A$ of length $n$ contains a Hamiltonian path.
\end{theorem}
\begin{proof} 
    Let $v_1$ be the index of the minimum value element of $A$.
    Then let $v_{i}$ be the smallest greater valued index than $v_{i - 1}$ for all $1 < i \leq n$.
    Since $G$ is a complete corresponding graph and since for every $i$ we have $A[v_{i - 1}] < A[v_i]$, the edge $v_{i - 1}v_i$ is in $G$.
    Thus the path $P = (v_1, \ldots, v_n)$ is in $G$ and spans all vertices of $G$.
    Hence $P$ is a Hamiltonian path in $G$.
\end{proof}

Let's show something about our particular case in the trivial algorithm.
\begin{proposition}
    If $G$ is a corresponding graph of an array $A$ with $n$ elements and reach $\lfloor n / 2 \rfloor$, then $G$ is a complete corresponding graph.
\end{proposition}
\begin{proof} 
    By way of contradiction assume that there exists an arc $e = (i, j)$ that we can add to $G$ where $i$ and $j$ are indices.
    Then there are at most $\lfloor n / 2 \rfloor$ elements in between $i$ and $j$ in $A$ since there are $n - 1$ elements besides $i$ or $j$.
    Thus $e$ is in $G$ which is a contradiction.
    Hence $G$ is a complete corresponding graph.
\end{proof}

Although the design of this algorithm is simple to follow, let's formally prove the correctness of the trivial algorithm with the following theorem.
\begin{theorem}
    Let $A' = \textsc{TrivialGraphSort}(A)$.
    Then $A'$ is the sorted array of $A$.
\end{theorem}
\begin{proof}
    In \textsc{TrivialGraphSort}, we first construct a corresponding graph $G$ of $A$ of $n$ elements with reach of $\lfloor n / 2 \rfloor$.
    Proposition 1 guarantees that $G$ is a complete corresponding graph.
    Then by Theorem 6, we know $G$ contains a Hamiltonian path.
    Then we run DFS starting at the minimum valued element which is an end-vertex of the Hamiltonian path, which gives us the topological sort which is the sequence of indices for the sorted array of $A$ by Theorem 5.
    Then we get the actual array from \textsc{ToArray} which is necessarily the sorted array $A'$ of $A$ as claimed.
\end{proof}

Thus we have completed showing that this algorithm indeed works where we sort an array using a corresponding graph and its topological sort.

\subsection{Run-time}
First we have $r = \lfloor n / 2 \rfloor$, so to construct the graph we have a run-time of $O(n (\lfloor n / 2 \rfloor)^2)$ by a previously stated formula, which equates to $O(n^3)$.
Next to find the minimum, our function has a run-time of $\Theta(n)$.
We have size
\begin{align*}
    m &= (n - 1) + (n - 2) + \cdots + (1) \\
    &= \sum_{k = 1}^{n - 1} (n - k) \\
    &= n(n - 1) - \frac{n(n - 1)}{2} \\
    &= \frac{n(n - 1)}{2}
\end{align*}
for $G$.
But, since $G$ contains a Hamiltonian path and we start at the minimum value vertex, we actually have a run-time of $\Theta(n)$ for DFS in this case.
Lastly to convert to an array it has a run-time of $\Theta(n)$ obviously.
Therefore, in total we have a run-time of
\begin{align*}
    O(n^3) + \Theta(n) + \Theta(n) + \Theta(n) = O(n^3)
\end{align*}
which is extremely inefficient since we also know the algorithm runs in time $\Omega(n^2)$ obviously too.
However, this algorithm mainly provides us with base strategy that we can employ further in the paper.

\subsection{Example}
Let's review an example of the above described algorithm.
Let,
\begin{align*}
    A = [3.5, 2, 9, 11, 1, -2.2, 5].
\end{align*}
Then our corresponding graph is,
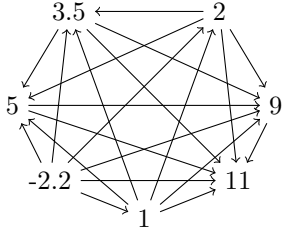
\begin{center}
    \begin{tikzpicture}
        \node (1) at (-1, 0) {3.5};
        \node (2) at (1, 0) {2};
        \node (3) at (1.75, -1.25) {9};
        \node (4) at (1.25, -2.25) {11};
        \node (5) at (0, -2.75) {1};
        \node (6) at (-1.25, -2.25) {-2.2};
        \node (7) at (-1.75, -1.25) {5};

        \path[->] (2) edge (1);
        \path[->] (1) edge (3);
        \path[->] (1) edge (4);
        \path[->] (5) edge (1);
        \path[->] (6) edge (1);
        \path[->] (1) edge (7);
        \path[->] (2) edge (3);
        \path[->] (2) edge (4);
        \path[->] (5) edge (2);
        \path[->] (6) edge (2);
        \path[->] (2) edge (7);
        \path[->] (3) edge (4);
        \path[->] (5) edge (3);
        \path[->] (6) edge (3);
        \path[->] (7) edge (3);
        \path[->] (5) edge (4);
        \path[->] (6) edge (4);
        \path[->] (7) edge (4);
        \path[->] (6) edge (5);
        \path[->] (5) edge (7);
        \path[->] (6) edge (7);
    \end{tikzpicture}.
\end{center}
To construct we take at most $7 \cdot \lfloor 7 / 2 \rfloor \cdot 7 = 147$ comparisons to construct the 21 edges of the above graph since we must also account for the background work of adding to the adjacency list.
Now we run DFS on our graph as specified for the algorithm (starting at the minimum element).
The bold path represents the path of discovery:
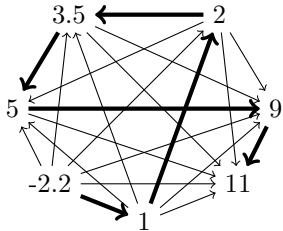
\begin{center}
    \begin{tikzpicture}
        \node (1) at (-1, 0) {3.5};
        \node (2) at (1, 0) {2};
        \node (3) at (1.75, -1.25) {9};
        \node (4) at (1.25, -2.25) {11};
        \node (5) at (0, -2.75) {1};
        \node (6) at (-1.25, -2.25) {-2.2};
        \node (7) at (-1.75, -1.25) {5};

        \path[->] (2) edge[ultra thick] (1);
        \path[->] (1) edge (3);
        \path[->] (1) edge (4);
        \path[->] (5) edge (1);
        \path[->] (6) edge (1);
        \path[->] (1) edge[ultra thick] (7);
        \path[->] (2) edge (3);
        \path[->] (2) edge (4);
        \path[->] (5) edge[ultra thick] (2);
        \path[->] (6) edge (2);
        \path[->] (2) edge (7);
        \path[->] (3) edge[ultra thick] (4);
        \path[->] (5) edge (3);
        \path[->] (6) edge (3);
        \path[->] (7) edge[ultra thick] (3);
        \path[->] (5) edge (4);
        \path[->] (6) edge (4);
        \path[->] (7) edge (4);
        \path[->] (6) edge[ultra thick] (5);
        \path[->] (5) edge (7);
        \path[->] (6) edge (7);
    \end{tikzpicture}.
\end{center}
To run DFS we witness that it takes only 6 edge and 7 vertex traversals since there exists a Hamiltonian path.
Then our DFS stack is our sorted array:
\begin{align*}
    A' = [-2.2, 1, 2, 3.5, 5, 9, 11].
\end{align*}
Thus we have sorted the array as desired using the trivial algorithm.
Notice that the most costly operation is constructing the graph which is bounded at a whopping 147 array comparisons.

\section{Properties of the DFS Forest}
The DFS algorithm will visit vertices in some fashion and the edges it uses to traverse the graph can be collected into a forest.
\begin{definition}
    Let $G$ be a graph, and let $F \subseteq G$ be a graph with $V(F) = V(G)$ and every edge $e = uv$ is in $F$ if and only if DFS on $G$ traverses that edge.
    Then $F$ is a resulting DFS forest of $G$.
\end{definition}
\begin{remark}
    We know every DFS forest is a forest because DFS visits every vertex once which allows for no cycles.
\end{remark}
\begin{remark}
    Note that if $G$ is a comparison graph, then the resulting DFS forest $F$ is also a comparison graph.
    In this way, the DFS forest of $F$ is also a comparison graph.
\end{remark}

It is assumed the reader has prior knowledge of DFS forests and forests in general.
Now let's prove a point that is the basis of our next algorithm.
\begin{theorem}
    Let $G$ be a comparison graph of array $A$, and let $F$ be the resulting DFS forest of $G$.
    Then any path in $F$ is a subsequence of indices from the sorted array of $A$ (not necessarily consecutive).
\end{theorem}
\begin{proof}
    Since $F \subseteq G$, every arc $e = uv$ in $F$ has the property that $A[u] < A[v]$.
    Let $P = (v_1, \ldots, v_k)$ be a path in $F$ where $v_i$ is a vertex for all $1 \leq i \leq k$.
    Then we have
    \begin{align*}
        A[v_1] < \ldots < A[v_k]
    \end{align*}
    which is necessarily sorted.
    Thus the path $P$ is a subsequence of indices from the sorted array of $A$.
\end{proof}

Note we will start abusing the notion of the connected component (or general component) by having it refer to the connected components of the underlying graph of some directed graph.
For our next theorem, let's first introduce some new definitions.
\begin{definition}
    A rooted tree is a tree with a distinguished vertex called the root such that the root has no incoming edges.
\end{definition}
\begin{remark}
    The root may also be the head of a Hamiltonian path in a comparison graph. This is just abuse of the traditional definition.
\end{remark}
\begin{definition}
    A leaf of a rooted tree is an end-vertex of the tree.
\end{definition}
\begin{definition}
    An internal vertex of a rooted tree is a vertex of the tree that is not the root or a leaf of the tree.
\end{definition}

\begin{remark}
    The connected components of a resulting DFS forest are all rooted trees.
\end{remark}

Now let's show an important feature of corresponding graphs of a reach one and their resulting DFS forests.
\begin{proposition}
    Let $G$ be a corresponding graph of array $A$ and reach one, and let $F$ be the resulting DFS forest of $G$.
    Let $T$ be any component of $F$. Then every root of $T$ has at most two outgoing arcs, every leaf of $T$ has at most one incoming arc, and every internal vertex of $T$ has one incoming and one outgoing arc.
\end{proposition}
\begin{proof}
    First note that every vertex of $G$ has a total degree of two, i.e. the sum of the in-degree and out-degree by definition of a reach of one.
    Then since $F \subseteq G$, every vertex of $F$ has a total degree of at most two.
    Let $T$ be any connected component of $F$.
    Then since, by definition, roots have an in-degree of zero, every root of $T$ can have at most two outgoing arcs.
    Since we visit every vertex once in DFS, no vertex in $T$ has an in-degree of two, and by definition a leaf has no outgoing arc, so it can have at most one incoming arc.
    Since by definition, an internal vertex in $T$ is not the root or a leaf of $T$, so it has at least one incoming and outgoing arc, and since our total degree is at most two, necessarily every internal vertex in $T$ has one incoming and one outgoing arc precisely.
\end{proof}
\begin{remark}
    Every rooted tree of $F$ has at most two sub-trees stemming from the root, and those sub-trees are paths.
\end{remark}

\section{The Merge Process}
We observed that in the trivial algorithm, the most expensive procedure was to construct the graph which took $O(n^3)$ time since the reach was dependent on $n$.
If $r$ is fixed at a constant value, the run-time for construction simply becomes $\Theta(n)$ since run-time is $O(nr^2)$.
However, when we have constant $r$, the corresponding graph does not necessarily contain a Hamiltonian path.
Hence, in this section we will develop a method of tackling this issue.

Let $G$ be the corresponding graph of array $A$ with constant reach $r$.
Then by DFS on $G$ we achieve a DFS forest $F$ which contains $k$ connected components.
Now let's say we merge the components in a way such that new graph $H$ and its resulting DFS forest $F'$, have $\lceil k / 2 \rceil$ connected components.
This idea will be more refined into a concrete algorithm in the next section; right now, we just want to develop such a merge process.

There are many ways we can approach this problem, but we will particularly attempt at merging pairs of components to generate a component that contains a Hamiltonian path.
This is easier to tackle with a reach of one since by Proposition 2, every tree in $F$ has at most two sub-trees stemming from the root.
Let a vertex $v$ be the root of a tree $T$ in $F$ such that $u$ and $w$ are the neighbors of $v$.
Then we cannot conclude anything about the order or values of the vertices in the sub-trees that stem from $u$ and $w$, call them $R$ and $S$.
Hence, we should merge sub-trees of trees where they exist before we start merging the components themselves.
By Proposition 2, both $R$ and $S$ are paths, and we wish to merge them into one graph $H$ that is also a comparison graph such that we conserve the comparison edge property and $H$ contains a Hamiltonian path.
It is important that $H$ contains a Hamiltonian path because then if we connect $v$ to the beginning of the path in $H$, we achieve a connected comparison graph with vertex set $V(T)$ and a Hamiltonian path.

Note we will use the term Hamiltonian path for individual components of a graph too from now on than the traditional definition for the entirety of a graph.

The merge process will take advantage of the fact that $R = (a_1, \ldots, a_x)$ and $S = (b_1, \ldots, b_y)$ are paths.
First we will compare $a_1$ to $b_1$ and add the corresponding arc between $a_1$ and $b_1$.
Without loss of generality, assume $a_1$ has lesser value than $b_1$.
Then we compare $a_2$ to $b_1$ and add their corresponding arc, and continue this process till we find some $a_i$ that has greater value than $b_1$ for some $1 < i \leq x$.
Then we compare $a_i$ to $b_2$ and continue the process till we reach the end of either $R$ or $S$.
Let the resulting graph of these operations be $H$.
Now we must show that $H$ contains a Hamiltonian path to satisfy our requirements for our merge process.
\begin{theorem}
    Let $T$ be a connected component of $F$, the resulting DFS forest of graph $G$.
    If $T$ contains two sub-trees stemming from the root, merge the two sub-trees using the above process, then the merged graph, call it $H$, contains a Hamiltonian path.
\end{theorem}
\begin{proof}
    We will show that $H$ contains a Hamiltonian path by construction.
    Let $v \in V(T)$ be the root of $T$, and let $u_1$ and $w_1$ be the two vertices adjacent from $v$.
    Let $P$ be our path to which we will add arcs.
    Without loss of generality, assume $u_1$ has lesser value than $w_1$, so add the edge $vu_1$ to $P$.
    Let $u_2$ be the vertex adjacent from $u_1$ that is not $w_1$ (must exist because $u_1$ and $u_2$ are part of the same sub-tree of $T$).
    Then if $u_2$ has lesser value than $w_1$, add the edge $u_1u_2$ to $P$, else add $u_1v_1$ to $P$.
    In this way we construct our path by adding an arc from the current vertex to an adjacent vertex of minimum value.
    These ``cross'' arcs between vertices exist since we have merged in a particular fashion which allows for their existence.
    Hence our path $P$ will eventually contain all the vertices, so it is a Hamiltonian path by definition in $H$ as claimed.
\end{proof}
\begin{remark}
    The construction of this path is analogous to how DFS would operate on $H$ with a sorted adjacency list representation of $H$.
\end{remark}

Thus we have merged the sub-trees of $T$ and achieved a new comparison graph that contains a Hamiltonian path.
Any tree in $F$ that has a root with only one outgoing vertex is a Hamiltonian path itself.
Thus the process of merging pairs of connected components boils down to merging paths, similar to above.
However it is important to know how to traverse only the Hamiltonian path of the connected components which is where we take advantage of the sorted property of the adjacency list.
We will now develop another merge process for a pair of components that contain a Hamiltonian path.

Let $H_1$ and $H_2$ be two components of $H$ where every component of $H$ contains a Hamiltonian path.
Although we can find the Hamiltonian paths of both graphs by running DFS on both individually, let's simply start at the minimum valued vertices of both $H_1$ and $H_2$, call them $a_1$ and $b_1$.
Now we compare the two and add a corresponding arc.
Without loss of generality, assume $a_1$ has lesser value than $b_1$.
Then let $a_2$ be the smallest value vertex adjacent from $a_1$ (similar to a recursive DFS visit), and compare $a_2$ to $b_1$ and add the corresponding arc.
Continue this process till we reach some $a_i \in V(H_1)$ that has value greater than $b_1$.
Then we compare $a_i$ to $b_1$ and add their corresponding arc, and continue in the same fashion with $b_j \in V(H_2)$.
We stop this process till we have no vertices adjacent from the current one in either $H_1$ and $H_2$.
We have essentially followed the DFS strategy but added our corresponding arcs in the process similar to the merging of the sub-trees.
Lastly we must show that this merge of $H_1$ and $H_2$ produces a graph that includes a Hamiltonian path.
\begin{theorem}
    Let $F$ be a resulting DFS forest of $G$.
    Let $H = F$ if all components of $F$ are paths, otherwise merge the sub-trees of $F$ and let $H$ be the result.
    Let $H_1$ and $H_2$ be a pair of two distinct components of $H$.
    Then if we run the above mentioned process, our resulting graph $H'$ contains a Hamiltonian path.
\end{theorem}
\begin{proof}
    If $H = F$, then $H_1$ and $H_2$ are Hamiltonian paths, so let $P_1 = H_1$ and $P_2 = H_2$.
    Otherwise, by Theorem 9, we know $H_1$ and $H_2$ contains Hamiltonian paths, call them $P_1$ and $P_2$.
    Then traverse the graph similar to how we merge, i.e. we only ``cross'' between the two graphs when the direction of edges switches between to the two graphs.
    Thus we get a path $P$ that starts at the minimum value vertex of both $H_1$ and $H_2$, and then from there we continue to the next minimum and so on.
    All these arcs exist due to our merge process.
    Hence our resulting graph $H'$ contains a Hamiltonian path as claimed.
\end{proof}

Now let's summarize the merge processes.
In the first, we run DFS on our corresponding graph $G$ with reach one, and let $F$ be the resulting forest then merge the sub-trees of the trees in $F$.
And the second process is to merge consecutive components of $H$ such that all the components of $H$ contain Hamiltonian paths to gain a new graph called $H'$ which contains approximately half number of components in $H$ as desired.
Also every component of $H'$ contains a Hamiltonian path.
Note in practice, $H$ will either be the DFS forest itself or the graph after the sub-tree merge process on the DFS forest.

\subsection{Pseudo-code}
In our case, we have two merge processes actually: one that merges sub-trees and one that merges the components.
Notice that both these processes simply merge on the Hamiltonian path starting at some root.
In the case of the sub-trees we let the two roots be the neighbors of the root of the component.
On the other hand, when we merge the components let the roots be the roots of the contained Hamiltonian paths.
Again we will use the adjacency list representation.
So the pseudo-code for this helper function is the following where $H$ is a modified version of the DFS forest, $F \subseteq H$ resulting from corresponding graph $G$ of array $A$.
Note $H$ is a comparison graph of array $A$.
\begin{algorithmic}
    \Function {Merge}{$H, A, x, y$}
        \State $u = x$, $v = y$
        \While {$u \neq \mathrm{nil}$ and $v \neq \mathrm{nil}$}
            \If {$A[u] < A[v]$}
                \State \Call {AddArc}{$H, A, u, v$}
                \State $u = \mathrm{adj}[u].\mathrm{front}$
            \ElsIf {$A[u] > A[v]$}
                \State \Call {AddArc}{$H, A, v, u$}
                \State $v = \mathrm{adj}[v].\mathrm{front}$
            \EndIf
        \EndWhile
    \EndFunction
\end{algorithmic}
Note $x$ and $y$ are the roots that we define for that specific tree.
When we traverse to the ``child'' of the current vertex we are essentially traversing to the smallest value vertex adjacent from the current one.

Before we discuss the main merge procedures, let's first find the components of some DFS forest, specifically their roots.
After running DFS on some graph $G$, all vertices in $G$ with no incoming vertices will have no parent from DFS.
Note the converse is also true since we start DFS from the minimum valued element which necessarily is not the terminus of any arc.
Thus all the vertices after DFS that have no parent are the roots of the components of the DFS forest.
The following is the pseudo-code for this algorithm given the resulting DFS forest $F$.
\begin{algorithmic}
    \Function {FindRoots}{$F$}
        \State $\mathrm{roots} = ()$
        \For {$v \in V(F)$}
            \If {$v.\mathrm{parent} = \mathrm{nil}$}
                \State \Call{Append}{$\mathrm{roots}, v$}
            \EndIf
        \EndFor
        \State \Return $\mathrm{roots}$
    \EndFunction
\end{algorithmic}
Note $\mathrm{roots}$ is a list.
Also note that the list is in no specific order; it is completely arbitrary.

Now that we have our helper functions ready, we can delve into the actual merge process for a DFS forest $F$.
First for every connected component of $F$ we will merge the sub-trees if needed, i.e. if the root of the component has a degree of 2.
\begin{algorithmic}
    \Function {MergeSubTrees}{$H, A$}
        \State $\mathrm{roots} = $ \Call{FindRoots}{$F$}
        \For {$r \in \mathrm{roots}$}
            \If {$\mathrm{adj}[r].\mathrm{length} = 2$}
                \State $x = \mathrm{adj}[r].\mathrm{front}$
                \State $y = \mathrm{adj}[r].\mathrm{back}$
                \State \Call {Merge}{$H, A, x, y$}
            \EndIf
        \EndFor
    \EndFunction
\end{algorithmic}
Note the roots of the components of $H$ remain unchanged.

Then we will merge consecutive components of $F \subseteq H$ where every component of $H$ contains a Hamiltonian path.
So we have the following psuedo-code for given $H$ and corresponding array $A$.
\begin{algorithmic}
    \Function {MergeTrees}{$H, A$}
        \State $\mathrm{roots} = $ \Call{FindRoots}{$F$}
        \For {$j = 1$ to $\lfloor \mathrm{roots}.\mathrm{length} / 2 \rfloor$}
            \State $p = \mathrm{roots}.\mathrm{get}(2j - )$
            \State $q = \mathrm{roots}.\mathrm{get}(2j)$
            \If {$A[p] > A[q]$}
                \State \Call{Delete}{$\mathrm{roots}, p$}
            \Else
                \State \Call{Delete}{$\mathrm{roots}, q$}
            \EndIf
            \State \Call{Merge}{$H, A, p, q$}
        \EndFor
        \State \Return $\mathrm{roots}$
    \EndFunction
\end{algorithmic}
We return the list of ``new roots'' in the merged graph $H$.
It's important we get this list for our next algorithm where we need the specific list and the number of ``roots.''
Note that by roots we mean the roots of the Hamiltonian path; a little abuse of terminology again.

\subsection{Run-time}
We will first analyze the run-time for the helper functions before analyzing the main merge process.
Let $T$ be a component of a DFS forest, and let $T$ be of order $n$.
Let $T_1$ and $T_2$ be two sub-trees of $T$.
Let their respective orders be $n'_1$ and $n'_2$ where $n = n'_1 + n'_2 + 1$ (where the extra addition is for the root).
To merge these two graphs according to our merge processes, we are essentially traversing the respective Hamiltonian paths which are of the same order as their respective graph.
Every vertex can produce at most one arc originating from it because after we create that arc we traverse to another vertex.
Hence, in the worst-case we produce an arc for every vertex (excluding the root), i.e. we produce $n'_1 + n'_2$ arcs.
Also since $T_1$ and $T_2$ have a maximum of one outgoing vertex (otherwise they would not be part of a sub-tree), the run-time to add an arc takes $\Theta(1)$.
Hence to add $n'_1 + n'_2$ arcs it takes time $O(n'_1 + n'_2)$ which is $O(n - 1)$ which is the run-time of the sub-tree merge.

Now let's analyze the merging of any two components that contain Hamiltonian paths.
Let $H_1$ and $H_2$ be these two components of order $n_1$ and $n_2$.
In the worst-case every vertex in $H_1$ and $H_2$ has at most two outgoing vertices, thus to add an arc will still take $\Theta(1)$.
Again the worst-case merge of these two components will be where we add an arc for every vertex, so we have a total run-time of $O(n_1 + n_2)$ in general.

To find the roots for a forest $F$ we simply have a run-time of $O(|V(F)|)$ since it is only one loop, so it's linear with respect to the order.

Now for the main merge processes.
Let $k$ be the number of components of some DFS forest $F$ of order $n$.
In the worst-case, every component (tree) has two sub-trees, so in total we have a run-time of $O(n - k)$ which is also the number of arcs in the original $F$.
Since $1 \leq k \leq n$, the merging of sub-trees takes $O(n)$ time.

Then we must merge consecutive components of some graph $H$ (where every component of $H$ contains a Hamiltonian path), and in the worst-case we merge to the fullest, i.e. we add every possible arc.
Since we merge consecutive trees, if we have $n = n_1 + \cdots + n_k$ for $n_i$ the order of component $i$ of $H$ ($1 \leq i \leq k$), then to merge all the trees is simply
\begin{align*}
    O(n_1 + n_2) + \cdots + O(n_{k - 1} + n_k) = O(n)
\end{align*}
assuming $k$ is even (equal asymptotically when $k$ is odd).
Hence, in general the merging of components takes time $O(n)$ as well.
Thus for both merge processes we take linear time to complete the procedures.

\subsection{Example}
Let's illustrate the merge process we have just discussed.
Consider the following forest with just two trees:
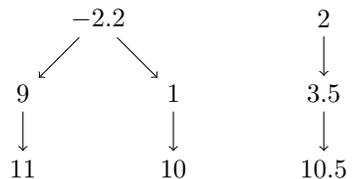
\begin{center}
    \begin{tikzpicture}
        \node (1) at (0, 0) {$-2.2$};
        \node (2) at (-1, -1) {9};
        \node (3) at (1, -1) {1};
        \node (4) at (-1, -2) {11};
        \node (5) at (3, 0) {2};
        \node (6) at (3, -1) {3.5};
        \node (7) at (3, -2) {10.5};
        \node (8) at (1, -2) {10};

        \path[->] (1) edge (2);
        \path[->] (1) edge (3);
        \path[->] (2) edge (4);
        \path[->] (3) edge (8);
        \path[->] (5) edge (6);
        \path[->] (6) edge (7);
    \end{tikzpicture}.
\end{center}
First we will merge the sub-trees:
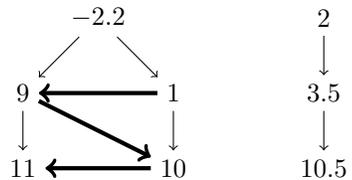
\begin{center}
    \begin{tikzpicture}
        \node (1) at (0, 0) {$-2.2$};
        \node (2) at (-1, -1) {9};
        \node (3) at (1, -1) {1};
        \node (4) at (-1, -2) {11};
        \node (5) at (3, 0) {2};
        \node (6) at (3, -1) {3.5};
        \node (7) at (3, -2) {10.5};
        \node (8) at (1, -2) {10};

        \path[->] (1) edge (2);
        \path[->] (1) edge (3);
        \path[->] (2) edge (4);
        \path[->] (3) edge (8);
        \path[->] (5) edge (6);
        \path[->] (6) edge (7);

        \path[->] (3) edge[ultra thick] (2);
        \path[->] (2) edge[ultra thick] (8);
        \path[->] (8) edge[ultra thick] (4);
    \end{tikzpicture}.
\end{center}
We add 3 more arcs to the forest to merge the sub-trees which is definitely under the upper bound of $8 - 2 = 6$ arcs (the number of vertices that are not roots).

We continue using the merged components from above; let's merge two components together:
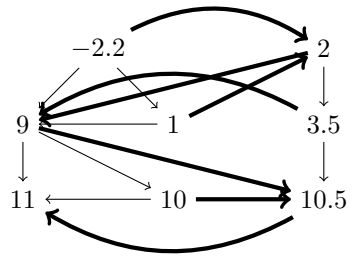
\begin{center}
    \begin{tikzpicture}
        \node (1) at (0, 0) {$-2.2$};
        \node (2) at (-1, -1) {9};
        \node (3) at (1, -1) {1};
        \node (4) at (-1, -2) {11};
        \node (5) at (3, 0) {2};
        \node (6) at (3, -1) {3.5};
        \node (7) at (3, -2) {10.5};
        \node (8) at (1, -2) {10};

        \path[->] (1) edge (2);
        \path[->] (1) edge (3);
        \path[->] (2) edge (4);
        \path[->] (5) edge (6);
        \path[->] (6) edge (7);

        \path[->] (3) edge (2);
        \path[->] (2) edge (8);
        \path[->] (8) edge (4);

        \path[->] (1) edge[bend left, ultra thick] (5);
        \path[->] (3) edge[ultra thick] (5);
        \path[->] (5) edge[ultra thick] (2);
        \path[->] (6) edge[bend right, ultra thick] (2);
        \path[->] (2) edge[ultra thick] (7);
        \path[->] (8) edge[ultra thick] (7);
        \path[->] (7) edge[bend left, ultra thick] (4);
    \end{tikzpicture}.
\end{center}
We add 7 more arcs to our graph which again satisfies our set upper bound of $8 - 1 = 7$ arcs (the number of vertices subtracted by 1).
Thus the size of our graph went from 6 to $6 + 3 = 9$ and then eventually to $9 + 7 = 16$ arcs.
And these two steps complete the merge process for a forest of only two trees.
Of course with multiple trees we just merge consecutive pairs of trees, using this process.

\section{A Divide-and-Conquer Algorithm}
After the buildup of the merge process, it begs us to discuss a Divide-and-Conquer algorithm.
Most importantly, Theorem 10 suggests that if we can generate a graph through a sequence of merges, we will end up with a graph containing a Hamiltonian path, and Theorem 5 asserts that the topological sort of that graph is in fact the sequence of indices of the sorted array.
Thus our goal is to create an algorithm that generates such a sequence of merges.

Let $A$ be an array $n$ elements which we wish to sort.
Also let $G$ be a corresponding graph of an array $A$ with reach one.
Let $F$ be the resulting DFS forest of $G$, and let $F$ have $k$ components which is in the range $1 \leq k \leq n$.
We first merge the sub-trees of the components in $F$, and then we run the merge process on $F$ to achieve a graph $H_1$ with $\lceil k / 2 \rceil$ components, and note each component of $H_1$ contains a Hamiltonian path (locally).
Then we run DFS on $H_1$ and let $F_1$ be the resulting forest which should also have $\lceil k / 2 \rceil$ components (we run DFS on the roots of the components of $H_1$).
Again we run the merge process on $F_1$ and get $H_2$ with $\lceil k / 2^2 \rceil$.
We continue this process till we achieve an $F_i$ such that $F_i$ is a tree, i.e. we have only one component.
Let $S$ be the topological sort of $F_i$ which will be our sequence of indices from the sorted array since $F_i$ is a Hamiltonian path (shown in the proof of correctness).
Therefore, to sort the array using our algorithm it takes $i$ merges, i.e. the number of components in $F_i$ is exactly one.

This algorithm is definitely a Divide-and-Conquer algorithm.
The division is the first run of DFS on the corresponding graph of the array.
The combine and conquer part is the merge process which reduces the trueness of our comparison graphs till we reach a trueness of one, i.e $\tau(H_{i - 1}) = \tau(F_i) = 1$.
Notice however that the number of merges depends on the number of components we have in $F$, the DFS forest of $G$, not the order of $G$.
Yet $k$ is bounded by $n$, so in the worst case we may achieve $n$ components.
Hence this algorithm can perform less computations for certain graph distributions; we will discuss the best-case run-time later in this section.

\subsection{Pseudo-code}
First we construct our corresponding graph $G$ of array $A$ of length $n$ and reach of one.
Then we must run DFS for the first time to get the resulting forest $F$ (we assume DFS replaces the given graph, so technically $G = F$).
Note to run DFS here, the order we visit the provisional roots can be arbitrary for the first time.
However, for the latter runs we must define the visit sequence by the roots, so that in this way we actually find the components of the graph.
Then we merge the sub-trees of $F$, and then merge the trees of $F$.
We continue this process till we have one component only, i.e. a trueness of one (shown in proof of correctness), which is when we stop, and then we run DFS for the last time to get the topological sort.
The following is the pseudo-code for this algorithm.
\begin{algorithmic}
    \Function {GraphSort}{$A$}
        \State $n = A.\mathrm{length}$
        \State $G = $ \Call{ConstructGraph}{$A, n, 1$}
        \State $S = (1, \ldots, n)$
        \State \Call{DFS}{$G, S$}
        \State \Call{MergeSubTrees}{$G, A$}
        \State $\mathrm{roots} = $ \Call{FindRoots}{$G$}
        \While {$\mathrm{roots}.\mathrm{length} > 1$}
            \State $\mathrm{roots} = $ \Call{MergeTrees}{$G, A$}
            \State $S = \mathrm{roots}$
            \State \Call{DFS}{$G, S$}
        \EndWhile
        \State \Return \Call{ToArray}{$A, S$}
    \EndFunction
\end{algorithmic}
Note that our DFS algorithm will essentially replace the graph given with its resulting DFS forest.
The dynamics with the DFS stack remain the same as the trivial algorithm.
Also note that we need to merge the sub-trees only once (at the beginning).
We will explore why we do this in the next part where we prove for correctness.

\subsection{Correctness}
Before we prove for correctness let's first prove why we only merge the sub-trees at the beginning.
\begin{lemma}
    Let $G$ be a corresponding graph of an array $A$ with $n$ elements and reach of one.
    Let $F$ be the resulting DFS forest of $G$.
    Then we  merge the sub-trees and components of $F$, and let $H$ be the resulting graph.
    Now run DFS on $H$ with our visiting list as the roots of $H$, and let $F'$ be the resulting forest.
    Then, all the components of $F'$ are directed paths.
\end{lemma}
\begin{proof} 
    First from Theorem 9, we know that after merging the sub-trees of the components of $F$, the components of the resulting graph $H'$ contain a Hamiltonian path.
    Now let's merge the components of $H'$ and let the resulting graph be $H$.
    By Theorem 10, every component of $H$ contains a Hamiltonian path.
    Now let $R$ be the list of new roots, i.e. the roots of the Hamiltonian paths of components in $H$.
    Then if we run DFS on $H$ with our visit list as $R$, we discover each component by traversing down the Hamiltonian path since our adjacency list is sorted.
    Hence, our DFS forest $F'$ is a collection of disjoint paths which correspond to the Hamiltonian paths of the components of $H$.
    Therefore, all the components of $F'$ are directed paths as claimed.
\end{proof}
\begin{corollary}
    If $H'$ be the graph after merging the components of $F'$, then if we run DFS on $H'$ on the roots, the components of the resulting forest are all paths.
\end{corollary}
\begin{proof}
    By Theorem 10, we know that the components of $H'$ contain Hamiltonian paths, so if we run DFS on the roots of those paths, obviously the resulting DFS forest is a graph of disjoint paths as claimed.
\end{proof}
Hence, we see that sub-trees in a secondary DFS forest will never exist, so we may ignore merging them after the first run of DFS.

Also let's conclude some properties of sequences of merges on a particular forest along with DFS.
\begin{lemma}
    Let $F$ be a comparison forest with $k$ components where every component contains a Hamiltonian path.
    Then let $H$ be the resulting graph after running the merge process (at the beginning both the sub-tree merge and component merge and later just the component merge) on $F$.
    Also if $k > 1$, $H$ has fewer than $k$ components.
    In fact, for any $k$, $F$ has $\lceil k / 2 \rceil$ components.
\end{lemma}
\begin{proof} 
    First if $k = 1$, then after the merge process, our components cannot increase, so we still have 1 component in $H$.
    Also $\lceil 1 / 2 \rceil = 1$, so the claim holds for $k = 1$.

    Assume $k > 1$, and let the components of $F$ be $T_1, \ldots, T_k$.
    Then the merge process essentially merges the components $T_1$ and $T_2$, and then $T_3$ and $T_4$, and so on.
    If $k$ is even, the merge process will merge $T_{k - 1}$ and $T_{k}$ too.
    Thus the number of components in $H$, the resulting graph, will be exactly $k / 2$ which is equal to $\lceil k / 2 \rceil$ since $k$ is even.
    On the other hand, if $k$ is odd, the merge process will merge $T_{k - 2}$ and $T_{k - 1}$, but not $T_k$ with any graph.
    Hence we do $\frac{k - 1}{2}$ merges, so $H$ has $\frac{k - 1}{2} + 1$ components which is equal to $\lceil k / 2 \rceil$ since $k$ is odd.

    Lastly, assume $k > 1$.
    If $k$ is even, then obviously
    \begin{align*}
        \left\lceil \frac{k}{2} \right\rceil = \frac{k}{2} < k
    \end{align*}
    as required.
    If $k$ is odd, then
    \begin{align*}
        \left\lceil \frac{k}{2} \right\rceil = \frac{k + 1}{2} < k
    \end{align*}
    since $k > 1$, so $k / 2 > 1 / 2$.
    Thus it holds that $F'$ has fewer components than $F$ if $F$ has more than one component.
\end{proof}

From the previous lemma it follows immediately that we require a finite number of merges to reach one component.
\begin{corollary}
    There exists a finite sequence of merge processes and DFS runs such that from our original comparison forest $F$ (of corresponding graph $G$) we will get a resulting a comparison graph $T$ that will have one component, i.e. $T$ is a tree.
\end{corollary}
\begin{proof}
    Assume $F$ has more than one component, because otherwise $F$ is the tree $T$.
    Let $F$ have $k > 1$ components, and let $F_1$ be the resulting DFS forest of the merged graph of $F$.
    Let $k_1$ be the number of components of $F_1$.
    Then $k_1 < k$ by Lemma 3 since $k > 1$.
    Now generate the graphs $F_i$ for $i > 1$ in the same way, and let $k_i$ be the number of components of $F_i$.
    Then some $k_i = 1$ since the sequence $(k, k_1, \ldots)$ is a strictly decreasing sequence until some $k_i = 1$.
\end{proof}

Now let's prove the correctness of our Divide-and-Conquer algorithm with the following theorems.
\begin{theorem}
    Let $A$ be an array of $n$ elements, and let $G$ be the corresponding graph of $A$ with reach one.
    Then let $F$ be the resulting DFS forest of $G$.
    Then we merge the sub-trees of $F$ to get a comparison graph $H'$.
    Now let $H_1$ be the resulting graph of the merge process on $H'$, and let $F_1$ be the resulting DFS forest on $H_1$ with the visiting list as a list of the roots of $H_1$.
    Then let $H_i$ be the resulting graph by the merge process on the forest $F_{i - 1}$ for all $i > 1$, and let $F_i$ be the resulting DFS forest of $H_i$ with the visiting list as a list of the roots of $H_i$.
    Then we claim that there exists a finite $i$ such that $F_i$ is a Hamiltonian path, and that the topological sort of $F_i$ is in fact the sequence of indices of the sorted array of $A$.
\end{theorem}
\begin{proof}
    First by Lemma 2, we know that the components of $F_1$ are all paths.
    Now $H_2$ is the resulting graph of the merge process on $F_1$, and by Theorem 10 we know that the components of $H_2$ all contain a Hamiltonian path.
    Now we let the new roots of $H_1$ be the list $R$, and run DFS on $H_2$ with the visit list as $R$.
    By Corollary 4, the resulting DFS forest $F_2$ is a collection of disjoint paths, i.e. the components of $F_2$ are all paths.
    Now we run the same process with $F_2$ to generate $F_3$ and so on till some $F_i$ such that $F_i$ is a tree and that $F_{i - 1}$ contains more than one component.
    By Corollary 5, we verify the existence of such an $i$ is finite.
    By the above defined process for generating $F_j$ for some $1 < j \leq i$, all the components of $F_j$ are paths.
    Hence, $F_i$ is necessarily a path since it is connected as it has only one component which is necessarily a path.
    Therefore, $F_i$ is a Hamiltonian path, and necessarily the topological sort of $F_i$ is the sequence of indices of the sorted array of $A$ by Theorem 5.
\end{proof}
\begin{corollary}
    If $A$ is an array of $n$ elements, then $A' = \textsc{GraphSort}(A)$ is sorted.
\end{corollary}
\begin{proof}
    Since the process of the algorithm is laid out by Theorem 11, it immediately follows that we attain the sorted array of $A$ by calling \textsc{GraphSort}(A) as claimed.
    Also the loop in \textsc{GraphSort} will terminate, again by Theorem 11.
\end{proof}
This concludes our proof of correctness for our algorithm.
Note we have only shown that the loop in the algorithm will terminate, and at that instant we have sorted the array.
In the next section, we will discuss when that termination occurs.

\subsection{Run-time}
This analysis of the run-time will be a sequential analysis of the main procedures in \textsc{GraphSort}, and then combine them to gain a full upper bound on the run-time.

First recall that the run-time to construct a corresponding graph for an array of $n$ elements is $O(nr^2)$ where $r$ is the reach.
Since our reach is one, our construction takes time $O(n)$ asymptotically.

Recall that the run-time to merge sub-trees is $\Theta(n)$ in the worst-case where $n$ is the number of elements.
On the other hand, also recall that to merge the components of a comparison graph it takes time $\Theta(n)$ also.
Also to find the roots we have a time of $\Theta(n)$.

Now let's discuss the run-time of DFS.
For the first run of DFS we do not know what our corresponding graph looks like, so we use the general run-time of DFS which is $O(|V(G)| + |E(G)|)$ where $G = (V, E)$ is our corresponding graph.
Obviously we have $n$ vertices in $G$ because every element maps to a vertex.
Further, since we have a reach of one, we have an out-degree of one for all vertices, so by the Degree-Sum Formula, we conclude $|E(G)| = n$ also.
Hence the first run of DFS runs in time $O(2n)$ which is asymptotically equal to $O(n)$.

For the latter runs of DFS we assure that the given graphs contain Hamiltonian paths, and our visit list for DFS is a list of roots of those paths.
Hence we never do an unnecessary check for a vertex visited or not in DFS, so our run-time is $\Theta(n)$ only.

Now let's compute the number of times we go through the main loop in our algorithm.
In the worst-case, we can have $n$ components in our first resulting DFS forest (this occurs only when the array is in reverse order).
Then through our first merge process we generate the next forest with $\lceil n / 2 \rceil$, and then the next forest with $\lceil n / 2^2 \rceil$.
This halts once we have one component, i.e. after some $i$ iterations we must have,
\begin{align*}
    &\quad \lceil n / 2^i \rceil = 1 \\
    &\therefore i = \lceil \log_2(n) \rceil.
\end{align*}
Hence we are done merging and running DFS after $\lceil \log_2(n) \rceil$ iterations in the worst-case.

Now that we have computed the run-time of all the individual parts, it's time to combine them.
Prior to the loop we first generate the corresponding graph which has time $O(n)$.
Then we run DFS on that graph which runs in time $O(n)$.
Then we merge the sub-trees which also has run-time of $O(n)$.
And lastly to find the roots we run in $\Theta(n)$ too.
Thus all these parts have a total run-time of
\begin{align*}
    O(n) + O(n) + O(n) + \Theta(n) = \Theta(n)
\end{align*}
which is linear.
Then in the body of the loop, we take time $\Theta(n)$ to merge the components, and then $\Theta(n)$ to run DFS on them.
So in total, the body takes time
\begin{align*}
    O(n) + \Theta(n) = \Theta(n)
\end{align*}
which is also linear.
Also the loop in the worst-case runs $\lceil \log_2(n) \rceil$ times which is logarithmic.
And finally, after the termination of the loop we convert our stack into the sorted array which takes time $\Theta(n)$.
Thus in total, our Divide-and-Conquer algorithm runs in time
\begin{align*}
    \Theta(n) + (\lceil \log_2(n) \rceil) \cdot \Theta(n) + \Theta(n) = \Theta(n \log n)
\end{align*}
in the worst-case.
In fact, $\Omega(n \log n)$ is the asymptotic lower bound of any comparison based sorting algorithm for worst-case.
Hence our algorithm is asymptotically as efficient as any mainstream sorting algorithm.

In the best-case which is when we have a sorted array, our corresponding graph itself contains a Hamiltonian path, so by Theorem 5, we have already sorted the array essentially without the loop.
Hence we have a run-time of $\Theta(n)$ in the best-case.
Hence, in general our algorithm runs in time $O(n \log n)$.

Additionally, the exact run-time in terms of basic operations (comparisons), we have an approximate run-time (in terms of basic operations) of
\begin{align*}
    2n\log_2(n) + 4n
\end{align*}
with many other computational optimizations added to the existing algorithm which we won't cover extensively in this paper.
One of the most efficient sorting algorithms called \textsc{QuickSort} runs in approximately
\begin{align*}
    2n\ln(n) + 2\ln(n) - 4n
\end{align*}
time.
So our algorithm is $\log_2(e)$ times slower than \textsc{QuickSort} which is approximately $1.44\times$.
Thus our algorithm isn't computationally ground-breaking, although it can be inspiration for one.
However we will discuss later some advantages of such an algorithm in practical applications.

\subsection{Example}
Consider the same array from before:
\begin{align*}
    A = [3.5, 2, 9, 11, 1, -2.2, 5].
\end{align*}
Then, we first make our corresponding graph of reach one:
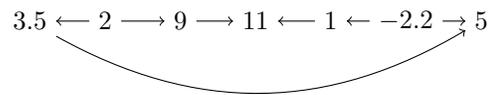
\begin{center}
    \begin{tikzpicture}
        \node (1) at (0, 0) {3.5};
        \node (2) at (1, 0) {2};
        \node (3) at (2, 0) {9};
        \node (4) at (3, 0) {11};
        \node (5) at (4, 0) {1};
        \node (6) at (5, 0) {$-2.2$};
        \node (7) at (6, 0) {5};

        \path[->] (1) edge[bend right] (7);
        \path[->] (2) edge (1);
        \path[->] (2) edge (3);
        \path[->] (3) edge (4);
        \path[->] (5) edge (4);
        \path[->] (6) edge (5);
        \path[->] (6) edge (7);
    \end{tikzpicture}.
\end{center}
We add 7 edges to construct the corresponding graph.
We will run DFS in the following order starting at $-2.2$, the minimum valued vertex.
Now we run DFS for the first time (bold lines denote path of discovery):
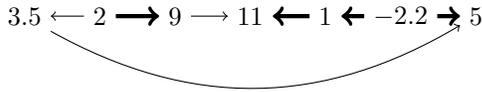
\begin{center}
    \begin{tikzpicture}
        \node (1) at (0, 0) {3.5};
        \node (2) at (1, 0) {2};
        \node (3) at (2, 0) {9};
        \node (4) at (3, 0) {11};
        \node (5) at (4, 0) {1};
        \node (6) at (5, 0) {$-2.2$};
        \node (7) at (6, 0) {5};

        \path[->] (1) edge[bend right] (7);
        \path[->] (2) edge (1);
        \path[->] (2) edge[ultra thick] (3);
        \path[->] (3) edge (4);
        \path[->] (5) edge[ultra thick] (4);
        \path[->] (6) edge[ultra thick] (5);
        \path[->] (6) edge[ultra thick] (7);
    \end{tikzpicture}.
\end{center}
We traverse 4 arcs and visit all 7 vertices by running DFS.
We generate a DFS forest with 3 trees with roots $(-2.2, 3.5, 2)$:
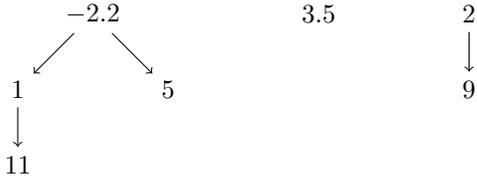
\begin{center}
    \begin{tikzpicture}
        \node (1) at (3, 0) {3.5};
        \node (2) at (5, 0) {2};
        \node (3) at (5, -1) {9};
        \node (4) at (-1, -2) {11};
        \node (5) at (-1, -1) {1};
        \node (6) at (0, 0) {$-2.2$};
        \node (7) at (1, -1) {5};

        \path[->] (5) edge (4);
        \path[->] (6) edge (5);
        \path[->] (6) edge (7);
        \path[->] (2) edge (3);
    \end{tikzpicture}.
\end{center}
After the merge process on the above forest, we get:
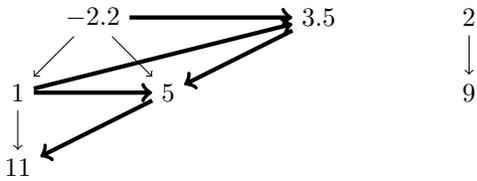
\begin{center}
    \begin{tikzpicture}
        \node (1) at (3, 0) {3.5};
        \node (2) at (5, 0) {2};
        \node (3) at (5, -1) {9};
        \node (4) at (-1, -2) {11};
        \node (5) at (-1, -1) {1};
        \node (6) at (0, 0) {$-2.2$};
        \node (7) at (1, -1) {5};

        \path[->] (5) edge (4);
        \path[->] (6) edge (5);
        \path[->] (6) edge (7);
        \path[->] (2) edge (3);

        \path[->] (5) edge[ultra thick] (7);
        \path[->] (7) edge[ultra thick] (4);

        \path[->] (6) edge[ultra thick] (1);
        \path[->] (5) edge[ultra thick] (1);
        \path[->] (1) edge[ultra thick] (7);
    \end{tikzpicture}.
\end{center}
We add 5 arcs to merge which is under $(7 - 3) + 7 = 11$ (the maximum number of arcs we can add) by the merge process.
So our ``new'' roots are $(-2.2, 2)$ since the root 3.5 got merged.
Now let's run DFS again starting at the new roots (note the list of roots being sorted is irrelevant):
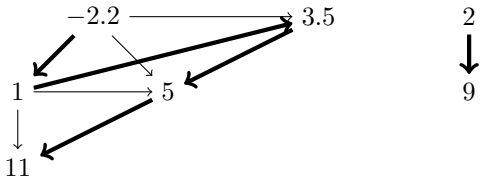
\begin{center}
    \begin{tikzpicture}
        \node (1) at (3, 0) {3.5};
        \node (2) at (5, 0) {2};
        \node (3) at (5, -1) {9};
        \node (4) at (-1, -2) {11};
        \node (5) at (-1, -1) {1};
        \node (6) at (0, 0) {$-2.2$};
        \node (7) at (1, -1) {5};

        \path[->] (5) edge (4);
        \path[->] (6) edge[ultra thick] (5);
        \path[->] (6) edge (7);
        \path[->] (2) edge[ultra thick] (3);

        \path[->] (5) edge (7);
        \path[->] (7) edge[ultra thick] (4);

        \path[->] (6) edge (1);
        \path[->] (5) edge[ultra thick] (1);
        \path[->] (1) edge[ultra thick] (7);
    \end{tikzpicture}.
\end{center}
Here DFS traverses 5 arcs and visits 7 vertices.
Thus we get the following DFS forest with 2 trees (paths in this case):
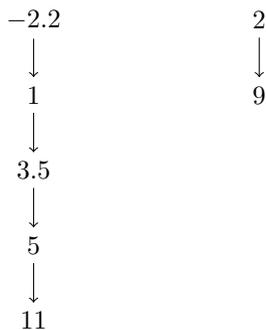
\begin{center}
    \begin{tikzpicture}
        \node (1) at (0, -2) {3.5};
        \node (2) at (3, 0) {2};
        \node (3) at (3, -1) {9};
        \node (4) at (0, -4) {11};
        \node (5) at (0, -1) {1};
        \node (6) at (0, 0) {$-2.2$};
        \node (7) at (0, -3) {5};

        \path[->] (6) edge (5);
        \path[->] (5) edge (1);
        \path[->] (1) edge (7);
        \path[->] (7) edge (4);
        \path[->] (2) edge (3);
    \end{tikzpicture}.
\end{center}
Notice they both are Hamiltonian paths.
Then again we run the merge process on this forest and our ``new'' root is $(-2.2)$ which is the minimum of the array also:
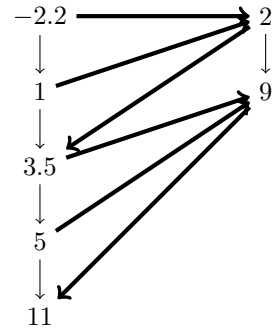
\begin{center}
    \begin{tikzpicture}
        \node (1) at (0, -2) {3.5};
        \node (2) at (3, 0) {2};
        \node (3) at (3, -1) {9};
        \node (4) at (0, -4) {11};
        \node (5) at (0, -1) {1};
        \node (6) at (0, 0) {$-2.2$};
        \node (7) at (0, -3) {5};

        \path[->] (6) edge (5);
        \path[->] (5) edge (1);
        \path[->] (1) edge (7);
        \path[->] (7) edge (4);
        \path[->] (2) edge (3);

        \path[->] (6) edge[ultra thick] (2);
        \path[->] (5) edge[ultra thick] (2);
        \path[->] (2) edge[ultra thick] (1);
        \path[->] (1) edge[ultra thick] (3);
        \path[->] (7) edge[ultra thick] (3);
        \path[->] (3) edge[ultra thick] (4);
    \end{tikzpicture}.
\end{center}
We add 6 arcs to our graph which is at the upper bound of $7 - 1 = 6$ (number of vertices subtracted by 1).
Now we run DFS for the last time starting at the root $-2.2$ (the minimum value vertex):
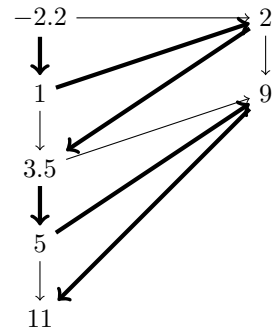
\begin{center}
    \begin{tikzpicture}
        \node (1) at (0, -2) {3.5};
        \node (2) at (3, 0) {2};
        \node (3) at (3, -1) {9};
        \node (4) at (0, -4) {11};
        \node (5) at (0, -1) {1};
        \node (6) at (0, 0) {$-2.2$};
        \node (7) at (0, -3) {5};

        \path[->] (6) edge[ultra thick] (5);
        \path[->] (5) edge (1);
        \path[->] (1) edge[ultra thick] (7);
        \path[->] (7) edge (4);
        \path[->] (2) edge (3);

        \path[->] (6) edge (2);
        \path[->] (5) edge[ultra thick] (2);
        \path[->] (2) edge[ultra thick] (1);
        \path[->] (1) edge (3);
        \path[->] (7) edge[ultra thick] (3);
        \path[->] (3) edge[ultra thick] (4);
    \end{tikzpicture}.
\end{center}
Here since the graph contains a Hamiltonian path we have 6 edge traversals and visit 7 vertices.
Also our DFS stack is our sorted array as desired:
\begin{align*}
    A' = [-2.2, 1, 2, 3.5, 5, 9, 11].
\end{align*}
Thus, we complete our Divide-and-Conquer algorithm.
Although visually the process is more extensive, its run-time is much more efficient than the trivial algorithm.

\section{Resemblance of MergeSort}
The \textsc{MergeSort} algorithm is a quintessential Divide-and-Conquer algorithm, and we will explore some similarities between \textsc{MergeSort} and \textsc{GraphSort}.

Firstly the merge process for our algorithm is so similar to \textsc{MergeSort}'s merge process, except in \textsc{GraphSort} we traverse a path, whereas in \textsc{MergeSort} we traverse sub-arrays.
However, one big difference is that in \textsc{MergeSort} we equally divide the array into sub-arrays of length one and build them up from there.
In \textsc{GraphSort}, we generate a graph with a certain reach and then after running DFS our ``building blocks'' are essentially the components of the DFS forest.
Our components have no fixed size which differs from \textsc{MergeSort} which is very organized and structures.
This is in fact an advantage of using graphs since we have to perform less unnecessary overhead operations.
Lastly, in \textsc{MergeSort} we keep an invariant on sub-arrays to always be sorted after a merge, and in our algorithm the invariant is essentially Theorem 8.
This is in fact equivalent since we are dealing with comparison graphs which are essentially graphical depictions of arrays which is what our corresponding graph tries to accomplish actually.

These are the most apparent and relevant similarities between the two.
In fact, we see that we can mimic \textsc{MergeSort} with our graphical strategy.

\subsection{A Graph Version of MergeSort}
In \textsc{MergeSort} we build up from sub-arrays of length one, then two, then four, etc.
So let's generate our corresponding graph as such: we generate an arc between consecutive pairs of elements using our same comparison property.
In this way every component after the first run of DFS has exactly 2 vertices (except one if we have an odd number of vertices).
Thus we have $\lceil n / 2 \rceil$ components exactly in our first DFS forest.
Then we can continue the original \textsc{GraphSort} algorithm with the same merging process to sort the given array.
Notice the only difference is that our components is fixed at $\lceil n / 2 \rceil$ rather than an arbitrary $k$ in the range $1 \leq k \leq n$.
Let's call this algorithm \textsc{GraphMergeSort}.

\subsection{Run-time}
The worst-case for our algorithm was when we had $n$ components in our first DFS forest, but with \textsc{GraphMergeSort}, we guarantee $\lceil n / 2 \rceil$ for any instance, i.e. we guarantee only one less iteration of the loop in the worst-case.
Hence our total run-time adapted from the run-time analysis for our algorithm is
\begin{align*}
    \Theta(n) + (\lceil \log_2(n) \rceil - 1) \cdot \Theta(n) + \Theta(n) = \Theta(n \log n).
\end{align*}
Hence, \textsc{GraphSort} and \textsc{GraphMergeSort} have equal efficiency asymptotically, but one just guarantees a certain number of iterations of the loop for our algorithm and one depends on an arbitrary integer at most $n$.

\subsection{Example}
Lastly, we analyze the \textsc{GraphMergeSort} algorithm which merely differs in the first corresponding graph construction. First consider the following array again:
\begin{align*}
    A = [3.5, 2, 9, 11, 1, -2.2, 5].
\end{align*}
Then our corresponding graph is the following:
\begin{center}
    \begin{tikzpicture}
        \node (1) at (0, 0) {3.5};
        \node (2) at (1, 0) {2};
        \node (3) at (2, 0) {9};
        \node (4) at (3, 0) {11};
        \node (5) at (4, 0) {1};
        \node (6) at (5, 0) {$-2.2$};
        \node (7) at (6, 0) {5};

        \path[->] (2) edge (1);
        \path[->] (3) edge (4);
        \path[->] (6) edge (5);
    \end{tikzpicture}.
\end{center}
We only add 3 arcs this time which is approximately half the number of array elements.
Then we continue the same process as for the previous algorithm, and again with our last run of DFS, the stack gives us the sorted array:
\begin{align*}
    A = [-2.2, 1, 2, 3.5, 5, 9, 11].
\end{align*}
This completes the \textsc{GraphMergeSort} example.

\section{Practical Implementation}
In this section we will focus less on the theory, but more on the implementation of the theory discussed throughout this paper.
This is imperative to convey since graphs can be represented by many data structures, so understanding which methods are efficient.
Furthermore, we can implement multiple computational optimizations to reduce the memory load and operations required.
These were mostly skipped in the development of the theory aspect since they add unnecessary complexity and do not affect the conclusion vastly.
We will also discuss how we can tackle equal-valued elements in an array which hasn't been discussed yet since we have assumed distinct values up till now for the ease in proving correctness.
Also, we will explore some machine dependent issues that we could face and how we can tackle them.
Lastly we will examine practical applications of this algorithm and what advantages it brings to the table relative to other efficient sorting algorithms out there.

\subsection{Data Structures}
As we have used throughout the pseudo-code sections of the paper, we will use the adjacency list representation of a graph.
This is imperative since we traverse through vertices rather than probing on edges as a whole, so quickly accessing neighboring vertices is important.
Moreover, we require our adjacency list to be sorted according to corresponding value for both our algorithms.
We can implement this sorted invariant as shown before by how we add arcs to the graph.

Graphs are a combinatorial mathematical structure, so it is important we remain memory efficient in our representations.
Since every vertex has a limited out-degree for our Divide-and-Conquer algorithm for the least, an adjacency list for every vertex is thus more memory efficient.
We take $\Theta(n + m)$ space to represent a graph of order $n$ and size $m$.

\subsection{Memory Optimizations}
Following our focus on memory efficiency, when we run DFS for the Divide-and-Conquer algorithm, we have DFS generate our forest and replace the given graph with the resulting forest.
However, we do not necessarily need to run secondary DFS operations; we included it primarily to reduce complexity in proofs.
Since when we merge we essentially traverse the contained Hamiltonian paths of the components, any overhead of arcs in the component will not affect the merge traversal.
Thus, we eliminate the middle uses of DFS; we still need to run DFS on the corresponding graph and at the end to gain the topological sort.
Note this optimization comes with a grain of salt, since now our graph data structure will be larger in terms of memory, but we limit alterations in the data structure as we use only one rather than many (the intermediate forests).

In order to optimize the construction of the forest separately, we can also implement a strategy where we remove edges from the given graph while running DFS.
We won't go into detail for a game plan to accomplish this, but leave it to the reader to conceive.
In the same discussion of space, we can alter our method of adding such that we delete those unnecessary arcs simultaneously.
This can be done in situations where have a list of consecutive vertices that all create an arc to one vertex in the other component when we are ``moving'' down a component in the process of merging.
So instead of our process of generating an arc and then ``moving,'' we simply compare the next vertex we would ``move'' to with the vertex we would add an arc to, and if we know we will ``move,'' we skip the addition of that arc.
This will reduce the number of ``cross''-arcs we have after merging, but it is much harder to also reduce the number of original arcs that already existed but are deemed unnecessary after merging.
Moreover, this increases the number of comparisons to be made during merging too.
We leave the implementation of a solution to the reader for the latter more difficult part.

\subsection{Introducing Randomness}
In \textsc{QuickSort} if our pivot choice is not random, let's say we choose the last element, then it's easy to construct an array that will force \textsc{QuickSort} to run in worst-case time.
However we tackled this by choosing a random pivot element.
Similarly, our algorithm highly depends on the number of components in our first DFS forest.
As we saw in the worst-case we may end up with a situation of $n$ components.
Our algorithm has a DFS visit ordering that is arbitrary (in the first run), and for simplicity we just go from the first element to the last in order.
However, if we implement a completely random order we may achieve a more consistent run-time overall.

If we assume that a random visit ordering implies an equal probability that we have either $k = 1, \ldots, n$ components, then in the average-case we have $n / 2$ components which is similar to our MergeSort version.
However this implication is probably false since DFS will first visit all the reachable vertices from where it started, and the fashion it reaches those vertices and which vertices are visited in the process can alter the upcoming DFS visits.
Hence we save this uncertainty with randomness and the average-case run-time analysis for another time.

\subsection{Dealing With Equal Value Elements}
One of the biggest issues we have yet to discuss is what does the algorithm do for equal value elements.
All the theorems and algorithms in the paper as of now assume that our given array contains only distinct values.
This is important to satisfy our comparison property which is a strict inequality.
In the case of equal-valued elements, we compare the indices of those elements in the given array to determine which direction an arc between the two elements will go.
Since indices are necessarily distinct, we will never face an issue there, and having this second condition for equal-valued elements will still produce an order relation for our comparison property.
Note this issue is purely computational and does not hinder with the basic algorithmic process or theorems; we simply modify our ordering definition for a comparison graph.

\subsection{A System Bottleneck}
An issue witnessed during an implementation of this algorithm on massive arrays was that DFS would fail midway.
This is because DFS is recursive, and from the second DFS run and onward (in the Divide-and-Conquer algorithm) we are running DFS on disjoint paths essentially which leads to very deep recursions which may lead to a stack overflow on machines with limited memory.
We will discuss an iterative version of DFS solely for scenarios where we are dealing with forests of components that contain Hamiltonian paths, and in fact we notice that we can translate this for our first DFS run.

\subsection{An Iterative DFS Solution}
We will now discuss an iterative method of DFS for our particular case.
We assume that we are running DFS on a comparison graph $G$ such that every component of $G$ contains a Hamiltonian path.
Also assume we have a set of the roots of those Hamiltonian paths called $R$.
Recall that DFS will visit some start vertex, and then visit a vertex adjacent to the start, and so on.
Once all neighbors of the current vertex have already been visited, we ``back-track'' to its parent and check for the same.
Once we have reached the start vertex, and we have no adjacent unvisited vertices, we stop DFS starting from that root.
Then we let our new start be the next unvisited vertex in our visit list given and continue the same process till we visit all the vertices.

The biggest problem with this algorithm is that with our convention of the adjacency list, it becomes unnecessary to back-track as we have visited all the vertices in a component once we have the need to back-track since every component of $G$ contains a Hamiltonian path.
Therefore, an iterative solution would be to start at some vertex $x \in R$ and then continue to the first vertex $y$ in the sorted adjacency list of $x$, and then continue the first vertex $z$ in the sorted adjacency list of $y$, until we reach a vertex $w$ where $w$ has no adjacent vertices.
Then it is sufficient to conclude that we have discovered the Hamiltonian path starting at $x$ for the component since $G$ has no back-edges since it is a comparison graph.
We run this iterative process for all vertices in $R$, and the paths we generate are equal to the DFS forest as required.

Note we can only do this when we know every component contains a Hamiltonian path, so this can replace all the DFS runs after the first.
This can also translate to our first DFS run since we have shown that only roots of the trees in the resulting DFS forest may have two vertices adjacent to them.
Then we can start this iterative process whenever we start discovering the respective provisional sub-trees which we have shown are paths.
Thus we eliminate the recursive nature of DFS for our particular algorithm which allows for more versatility on machines with limited stack sizes.

\subsection{Graphical Approach Advantages}
During the discussion of the run-time of our algorithms, it was mentioned that computationally, our algorithm still lacks to perform as well as \textsc{QuickSort}, the leading sorting algorithm right now.
Also memory-wise, we construct a separate data structure to sort which adds extra overhead and cost.
However, there do exist some advantages of using the graphical approach.

Firstly, if there is a situation where we wish to sort a comparison graph itself, which matches our algorithm's objective essentially, we can employ our merge techniques and DFS to achieve a more true comparison graph than the original.
This eliminates the process of converting a graph into a linear structure to sort using a general sorting algorithm.
Now we do not need to transform the input as we can modify the graph itself.
Many applications use comparison graphs and directed graphs to represent networks and what not, and our algorithm provides a way to better detail those graphs and rank the nodes of those networks in an efficient manner.

Secondly, if the array distribution contains long increasing sub-sequences, our algorithm can take full advantage of this feature which algorithms like \textsc{QuickSort} and \textsc{MergeSort} fail to achieve.
Furthermore, some tailing computations as in other sorting algorithms can be eliminated since an edge can concatenate two ``sub-arrays.''
Additionally, since our best-case run-time is $\Theta(n)$ which is when the given array is already sorted, our algorithm performs extremely well in partially sorted arrays, similar to \textsc{InsertionSort}.
A complete analysis of the run-time in terms of inversions is not present in this paper and is saved for another time.

Lastly, some procedures presented in this paper can provide a better representation for some linear structures like linked-lists.
We need not necessarily sort the input, rather we can provide those intermediate forests if they are sufficient for a user.
This is useful in an ongoing insertion situation where more elements are being inserted and keeping the structure somewhat sorted is important.
Then at the end we may complete the algorithm and sort the array.
Note that the trueness of the forest will be approximately the same for every insert since fairly quickly we can choose where to insert the element.

\section{Similar Algorithm Ideas}
We will now discuss some ideas that may improve the algorithms discussed in this paper, and what parameters we can tweak to achieve vastly different results and uncover new problems.

First we can modify the reach values for our corresponding graphs.
In \textsc{GraphSort} we define a fixed reach of one, but there may be different implementations of algorithms with different reach values.
The higher our reach, the more complex our corresponding graph is, and the harder it is to parse and process the graph as graphs are combinatorial structures.

Additionally, we can implement a $k$-way merge process instead of our 2-way merge.
It doesn't seem to better asymptotic run-time in some preliminary analysis, but in terms of computation we may achieve more efficient algorithms because the logarithm base would be larger.
However, again implementing merges for multiple components turns into a very massive problem very quickly.

Moreover, we can expand on the idea of \textsc{GraphMergeSort} and instead of generating components by pairs of consecutive elements, we can generate components by triples of consecutive elements, or even higher $k$-tuples of consecutive elements.
In this way, we first solve each of the components in a specific way and then continue with our merge process.
This is a way to implement Divide-and-Conquer with the corresponding graph itself.

There are many conventions we have set for our algorithms, and we have absolute freedom to experiment with different conventions to realize newer algorithms employing comparison graphs.

\section{Conclusion}
In summary, this paper explored properties of comparison graphs, corresponding graphs, topological sorts, and DFS to fuse together procedures and algorithms that solve the old-age sorting problem.
Our best algorithm ran in time $\Theta(n\log n)$ in the worst-case and $\Theta(n)$ in the best-case.
Hence, it is on par with mainstream sorting algorithms although the graphical strategy seems more complex.
Further, it serves extremely well for particular distributions of arrays, and even for less efficient distributions it's competitive with other such algorithms.

To remind you again, sorting an array is simply an application of the techniques provided here.
There are many other applications of our procedures available with comparison graphs, in particular.
Our merge process is probably the most important and eye-opening part of this paper as it gives us a way to generate comparison graphs that are more true.
It retains all the properties and conditions of its original graph and generates a new one that is more detailed and concrete, and this procedure can be applied for solving many other problems.

This is the direct outcome of employing graphs since they provide another layer of complexity and information.
The algorithms presented took advantage of these properties of graphs to develop an interesting method to sort although it may even render to other problems.

Hopefully, the algorithms, procedures, and theorems discussed may inspire you, the reader, to embark on a discovery of more applications and algorithms using the ideas discussed in this paper as a foundation.

\section*{Acknowledgements}
I would like to thank Professor Patrick Tantalo of the Computer Science and Engineering department at the University of California, Santa Cruz.
His lectures and discussions inspired me to look into this method and write this paper.

I would also like to thank Professor Seshadhri Comandur of the Computer Science and Engineering department at the University of California, Santa Cruz.
He mentored me throughout the latter stages of the paper to refine and polish it.
Without his help I would not be able to finish this paper.

Further, the teaching of Nathan Marianovsky of the Mathematics department at the University of California, Santa Cruz, in graph theory gave me the perfect arsenal to formulate many of these ideas formally.

Lastly, I would like to sincerely thank Dr. Rajat K. Pal from the University of Calcutta for sharing insights of his previous work in this line of research and for kindly assiting me through the publishing process.

\nocite{clrs, cz, nm, pt1, pt2}
\printbibliography

\end{document}